\documentclass[a4paper]{article}

\usepackage[T1]{fontenc}
\usepackage{fix-cm}

\usepackage{url}
\usepackage{hyperref}
\usepackage{authblk}
\usepackage{geometry}
\usepackage{amsmath,amssymb,amsthm}
\usepackage{array}
\usepackage[linesnumbered,ruled,vlined,noresetcount]{algorithm2e}
\DontPrintSemicolon

\SetCommentSty{mycommfont}
\SetArgSty{textnormal}
\SetKwInput{Notation}{Notation}
\SetKwInput{Predicates}{Predicates}
\SetKwProg{Fn}{function}{:}{end}
\SetKwInput{Variables}{Variables}
\SetKwInput{Parameter}{Parameter}
\SetKwInput{Note}{Note}
\SetKwInput{Remark}{Remark}
\SetKwInput{Initially}{Initially}

\usepackage{multirow}
\usepackage{xcolor}
\usepackage{xfrac}

\newcommand{\red}[1]{\textcolor{red}{#1}}

\newcommand{\violet}[1]{\textcolor{violet}{#1}}

\newtheorem{conjecture}{Conjecture}
\newtheorem{theorem}{Theorem}
\newtheorem{corollary}{Corollary}
\newtheorem{lemma}{Lemma}

\theoremstyle{definition}
\newtheorem{definition}{Definition}

\newtheorem{remark}{Remark}
\newtheorem{note}{Note}
\newtheorem{observation}{Observation}

\newcommand{\poly}{\operatorname{poly}}

\def\protocol#1{{\normalfont\textsc{#1}}\xspace}
\newcommand{\ssrk}{\calP_{\protocol{SSRK}}}
\newcommand{\ccd}{\protocol{Detect}}
\newcommand{\lsle}{\calP_{\protocol{LSLE}}}
\newcommand{\findtarget}{\protocol{FindTarget}}
\newcommand{\ranking}{\protocol{Rank}}
\newcommand{\phaseclock}{\protocol{PhaseClock}}

\newcommand{\coldbfull}{\protocol{Collision\allowbreak{}Detection\allowbreak{}WithBounds}}
\newcommand{\coldb}{\protocol{CDWB}}
\newcommand{\assignranks}{\protocol{AssignRanks}}

\newcommand{\calP}{\mathcal{P}}

\newcommand{\tif}{\textit{if }}
\newcommand{\telseif}{\textit{else if }}
\newcommand{\totherwise}{\textit{otherwise}}

\newcommand{\sinit}{s_0}

\newcommand{\ini}{a_{\mathrm{in}}}
\newcommand{\res}{a_{\mathrm{re}}}

\newcommand{\var}{\mathtt{var}}
\newcommand{\rst}{\mathtt{reset}}
\newcommand{\mode}{\mathtt{mode}}
\newcommand{\target}{\mathtt{target}}
\newcommand{\clock}{\mathtt{clock}}
\newcommand{\parity}{\mathtt{parity}}
\newcommand{\delay}{\mathtt{delay}}
\newcommand{\ind}{\mathtt{index}}
\newcommand{\rank}{\mathtt{rank}}
\newcommand{\leader}{\mathtt{leader}}
\newcommand{\nonce}{\mathtt{nonce}}
\newcommand{\cand}{\mathtt{cand}}
\newcommand{\detected}{\mathtt{det}}
\newcommand{\susp}{\mathtt{susp}}
\newcommand{\vlist}{\mathtt{list}}

\newcommand{\grid}{\mathtt{gID}}

\newcommand{\infectivity}{\mathtt{infectivity}}

\newcommand{\calC}{\mathcal{C}}

\newcommand{\cphase}{\calC_{\mathrm{phase}}}

\newcommand{\csafe}{\calC_{\mathrm{safe}}}

\newcommand{\calI}{\mathcal{I}}
\newcommand{\cinit}{\calC_{\mathrm{init}}}

\newcommand{\goalI}{\sfrac{n^2}{\rho} \cdot \log \rho}
\newcommand{\goalT}{\sfrac{n}{\rho} \cdot \log \rho}
\newcommand{\rtimeI}{n^{3/2} \log n}
\newcommand{\rtimeT}{\sqrt{n}\log n}

\newcommand{\constR}{c_R}

\newcommand{\constL}{c_L}
\newcommand{\constD}{c_D}
\newcommand{\constM}{c_M}
\newcommand{\constT}{c_T}
\newcommand{\cmin}{c_{\mathrm{min}}}

\newcommand{\modeF}{\mathsf{F}}
\newcommand{\modeD}{\mathsf{D}}
\newcommand{\modeR}{\mathsf{R}}
\newcommand{\ad}{A_{\modeD}}

\newcommand{\phase}{\mathit{Phase}}

\newcommand{\rnew}{r_{\mathrm{new}}}
\newcommand{\pini}{p_{1}}
\newcommand{\pres}{p_{2}}
\newcommand{\rini}{r_{1}}
\newcommand{\rres}{r_{2}}

\newcommand{\dres}{d_{2}}

\newcommand{\ncand}{n_{\mathrm{cand}}}

\newcommand{\mname}{\rho^2}
\newcommand{\rem}{m}
\newcommand{\rept}{\constT}
\newcommand{\maxnum}{\lfloor \rho \lg \rho \rfloor}

\newcommand{\tmax}{T_{4}}
\newcommand{\tmid}{T_{3.5}}
\newcommand{\rname}{R_{\mathrm{name}}}

\newcommand{\logn}{\lceil \lg n \rceil}

\title{
Complementary Time--Space Tradeoff for Self-Stabilizing Leader Election: Polynomial States Meet Sublinear Time
}
\date{}

\author{Yuichi Sudo}

\affil{Hosei University, Tokyo, Japan}

\begin{document}

\maketitle

\begin{abstract}
We study the self-stabilizing leader election (SS-LE) problem in the population protocol model, assuming exact knowledge of the population size $n$. Burman, Chen, Chen, Doty, Nowak, Severson, and Xu [BCC+21a] (PODC) showed that this problem can be solved in $O(n)$ expected time with $O(n)$ states. Recently, Gąsieniec, Grodzicki, and Stachowiak [GGS25] (PODC) proved that $n+O(\log n)$ states suffice to achieve $O(n \log n)$ time both in expectation and with high probability (w.h.p.). 
If substantially more states are available, sublinear time can be achieved.
The authors of [BCC+21] presented a $2^{O(n^\rho\log n)}$-state SS-LE protocol with a parameter $\rho$:
setting $\rho = \Theta(\log n)$ yields an optimal $O(\log n)$ time both in expectation and w.h.p., while $\rho = \Theta(1)$ results in $O(\rho\,n^{1/(\rho+1)})$ expected time. Recently, Austin, Berenbrink, Friedetzky, Götte, and Hintze [ABF+25] (PODC) presented a novel SS-LE protocol parameterized by a positive integer $\rho$ with $1 \le \rho < n/2$ that solves SS-LE in $O(\sfrac{n}{\rho}\cdot\log n)$ time w.h.p.\ using $2^{O(\rho^2\log n)}$ states. This paper independently presents yet another time--space tradeoff of SS-LE: for any positive integer $\rho$ with $2 \le \rho \le \sqrt{n}$, SS-LE can be achieved within $O\left(\sfrac{n}{\rho}\cdot \log\rho\right)$ expected time using $2^{2\rho\lg^2\rho + O(\log n)}$ states.
The proposed protocol uses significantly fewer states than [ABF+25] for any expected stabilization time above $\Theta(\sqrt{n}\log n)$.
When $\rho = \Theta\left(\sfrac{\log n}{\log^2 \log n}\right)$,
the proposed protocol is the first to achieve sublinear time while using only polynomially many states.
A limitation of our protocol is that the constraint $\rho\le\sqrt{n}$ prevents achieving $o(\sqrt{n}\log n)$ time, whereas the protocol of [ABF+25] can surpass this bound.  
\end{abstract}

\section{Introduction}
\label{sec:intro}
This paper investigates the self-stabilizing leader election (SS-LE) problem in the population protocol model introduced by Angluin, Aspnes, Diamadi, Fischer, and Peralta~\cite{AAD+06}, which has been extensively studied in the distributed computing community for over two decades. Recently, the time--space tradeoff for SS-LE, assuming exact knowledge of the population size~$n$, has drawn significant attention. We present a novel time--space tradeoff that substantially improves the best-known upper bound on space complexity required to achieve any target time complexity within the range $\Theta(\sqrt{n}\log n)$ to $o(n)$.

\subsection{Population Protocols}
In population protocols, we consider $n\ge2$ anonymous state machines called \emph{agents}, which together form a single \emph{population}. At each time step, an ordered pair of agents is selected uniformly at random to have a (pairwise) \emph{interaction}. The chosen agents assume the roles of \emph{initiator} and \emph{responder} and, following the transition function of a given protocol, update their states based on their current states and roles. Each agent outputs a symbol determined by its current state. Throughout this paper, $A$ denotes the set of all agents.

A \emph{protocol} is a five-tuple $\calP = (Q, \sinit, Y, \delta, \pi)$, where $Q$ is the set of agent states, $\sinit$ is the initial state (unused in self-stabilizing protocols that we will see later), $Y$ is the set of output symbols, $\delta\colon Q \times Q \to Q \times Q$ is the transition function, and $\pi\colon Q \to Y$ is the output function.\footnote{
We omit the input symbol set and the input function here, as the leader election involves no external inputs.
}
A global state of the population (or \emph{configuration}) is a function $C\colon A \to Q$ that represents the state of each agent. A configuration $C$ is \emph{output-stable} (or simply \emph{stable}) if no agent ever changes its output in any subsequent execution from $C$.

The complexity of a protocol is measured by its stabilization time and its number of states. 
The \emph{stabilization time} of a protocol $\calP$ is the number of interactions, divided by the population size $n = |A|$, required for the population under $\calP$ to reach a stable configuration.\footnote{This definition applies to static problems such as leader election and majority; a more intricate definition is required for non-static problems. See \cite{AAF+08}.} 
The quantity ``number of interactions divided by $n$'' is commonly called \emph{parallel time}, reflecting the fact that agents may interact concurrently in practice. 
Throughout this paper, ``time'' refers to parallel time. 
Since interacting agents are chosen uniformly at random, we evaluate stabilization time either \emph{in expectation} or \emph{with high probability}. 
We say that an event occurs \emph{with high probability (w.h.p.)} if, for any positive constant $\eta = \Theta(1)$, it occurs with probability $1 - O(n^{-\eta})$ (by appropriately adjusting the constant parameters of a protocol, if necessary). 
The \emph{number of states} of a protocol is $|Q|$, the cardinality of its state set $Q$. 
Because some protocols use super-exponentially many states, we also express space complexity in \emph{bits}: $x$ bits correspond to $2^x$ states.

\subsection{Related Work}
Time–space tradeoffs arise for many problems in population protocols: the more states available, the shorter the stabilization time. 
The tradeoff for (non-self-stabilizing) leader election was intensively studied between 2015 and 2020. 
In this problem, exactly one agent must output `L' (leader) and all others `F' (follower). 
Doty and Soloveichik~\cite{DS18} proved that any $O(1)$-state protocol requires $\Omega(n)$ expected time, establishing the time-optimality of the two-state protocol of \cite{AAD+06}, which stabilizes in $O(n)$ expected time. 
Alistarh and Gelashvili~\cite{AG18} gave a protocol with $O(\log^3 n)$ states and $O(\log^3 n)$ expected time, and subsequent work~\cite{AAE+17, BCER17, AAG18, BKKO18, GS20, GSU18, SOI+20} progressively improved these bounds. 
Finally, Berenbrink, Giakkoupis, and Kling~\cite{BGK20} obtained an $O(\log\log n)$-state and $O(\log n)$-expected-time protocol, which is optimal in both time and space. 
Optimality follows from two lower bounds:
$o(n^2/\poly(\log n))$ expected time requires at least $\tfrac{1}{2}\log\log n$ states~\cite{AAE+17}, and $\Omega(\log n)$ expected time is necessary~\cite{SM20}.\footnote{The latter bound goes beyond the simple coupon-collector argument,
since all agents may initially be in the follower state.}

This paper studies self-stabilizing leader election (SS-LE). Self-stabilization, introduced by Dijkstra~\cite{Dij74}, is a fault-tolerance concept that requires a protocol to converge to a correct configuration starting from an \emph{arbitrary} configuration. A population protocol $\calP$ solves SS-LE if, starting from any configuration, it eventually stabilizes with exactly one leader.
Solving SS-LE is more challenging than (non-self-stabilizing) leader election. An SS-LE protocol must create a leader when none exists and reduce multiple leaders to one; these opposing goals often conflict. Indeed, using the partition technique of \cite{AAF+08,CIW12}, one can easily show that SS-LE is impossible unless each agent knows the exact population size $n$ (see \cite{SOK+20}, p.~618).
Nevertheless, SS-LE has been extensively studied under additional assumptions or relaxed stabilization requirements.

One approach to SS-LE is to assume that each agent knows the exact population size $n$, i.e., $n$ is hard-coded into the protocol. 
Tables~\ref{table:silent} and~\ref{table:non-silent} summarize the upper and lower bounds in this setting, including our contributions, for \emph{silent} and \emph{non-silent} protocols. 
Silent protocols require the population to reach a configuration from which no agent ever changes its \emph{state}, a stronger notion than output stabilization. 
Note that expected and w.h.p.~stabilization times differ by at most an $O(\log n)$ factor.
By standard arguments for self-stabilizing protocols, the following holds:

\begin{remark}
For any function $f:\mathbb{N}^+\to\mathbb{N}^+$, 
if a protocol solves SS-LE in $O(f(n))$ expected time, 
then it solves SS-LE in $O(f(n)\log n)$ time w.h.p.; 
conversely, $O(f(n))$ time w.h.p.\ implies $O(f(n))$ expected time.
\end{remark}

All existing SS-LE protocols with exact knowledge of the population size $n$ (Tables~\ref{table:silent} and~\ref{table:non-silent}) also solve the more general problem of \emph{self-stabilizing ranking} (SS-RK).
\begin{definition}[Self-Stabilizing Ranking (SS-RK)]
A protocol~$\calP$ solves SS-RK if, starting from any configuration, the population under~$\calP$ reaches, with probability~1, an output-stable configuration in which all agents output distinct integers in $[1,n]$.
\end{definition}
\noindent
We call such outputs \emph{ranks}. 
Once unique ranks are assigned, the agent with rank $1$ can serve as the leader; hence SS-LE reduces to SS-RK, while the converse reduction remains open. 
Since all known SS-LE protocols with knowledge of $n$ also solve SS-RK, this motivates the following conjecture.

\begin{conjecture}
\label{conjecture:ss-rk}
SS-LE and SS-RK have the same time--space tradeoffs: for any functions $f,g:\mathbb{N}\to\mathbb{N}$,
SS-LE can be solved in expected (resp.\ w.h.p.) stabilization time $O(f(n))$ using $O(g(n))$ states
if and only if SS-RK can be solved in expected (resp.\ w.h.p.) stabilization time $O(f(n))$ using $O(g(n))$ states.
\end{conjecture}

Assuming knowledge of $n$, Cai, Izumi, and Wada~\cite{CIW12} proved that any SS-LE protocol requires at least $n$ states and presented a silent, state-optimal ($n$-state) protocol with $O(n^2)$ expected time. 
Burman, Chen, Chen, Doty, Nowak, Severson, and Xu~\cite{BCC+21} showed that increasing the number of states by a constant factor reduces the expected stabilization time to $O(n)$, and this is time-optimal for silent SS-LE: any silent protocol requires $\Omega(n)$ expected time (and $\Omega(n\log n)$ w.h.p.). 
Dropping the silent requirement enables sublinear time (i.e., $o(n)$ time). 
They presented an $O(\log n)$-time protocol w.h.p.\ using super-exponentially many states (i.e., $O(n^{\log n}\log n)$ bits). No faster protocol exists, 
which just follows from the coupon collector argument (in contrast to the non-self-stabilizing case~\cite{SM20}).
They further gave a protocol parameterized by $\rho=\Theta(1)$ achieving $O(\rho\,n^{1/(\rho+1)})$ expected time using $O(n^\rho\log n)$ bits.

\begin{table}[t]
\caption{
Silent SS-LE protocols with exact knowledge of the population size $n$ (\red{space shown in number of \emph{states}}); all protocols also solve SS-RK.
}
\label{table:silent}

\vspace{0.1cm}

\centering
\begin{tabular}{c c c c}
\hline
 protocol & expected time & w.h.p.~time & \red{\emph{states}}
 \\ 
\hline
\cite{CIW12} & $O(n^2)$ & $O(n^2)$ & $n$ \\
\cite{BCC+21} & $O(n)$ & $O(n \log n)$ & $O(n)$ \\
\cite{BEG+25} & $O(n \log n)$ & $O(n \log n)$ & $n+O(\log^2 n)$  \\
\cite{GGS25} & $O(n \log n)$ & $O(n \log n)$ & $n+O(\log n)$  \\
\cite{GGS25} & $O(n^{7/4} \log^2 n)$ & $O(n^{7/4} \log^2 n)$ & $n+1$  \\
\hline 
\cite{CIW12} & any & any & $\ge n$ \\
\cite{BCC+21} & $\Omega(n)$& $\Omega(n \log n)$  & any \\
\hline
\end{tabular} 

\vspace{0.5cm}

\caption{ 
 Non-silent SS-LE protocols with exact knowledge of the population size $n$
 (\red{space shown in \emph{bits}; $x$ bits correspond to $2^x$ states});
 all protocols also solve SS-RK.
}
\label{table:non-silent}

\vspace{0.1cm}

\centering
\begin{tabular}{c c c c c }
\hline
 protocol & expected time & w.h.p.~time & \red{\emph{bits}} & parameters \\ 
\hline 
\cite{BCC+21} & $O(\log n)$ & $O(\log n)$ & $O(n^{\log n}\log n)$ & - \\

\cite{BCC+21} & $O(\rho\cdot n^{1/(\rho+1)})$ & $O(\rho\cdot n^{1/(\rho+1) }\log n)$ & $O(n^{\rho}\log n)$ & 
$\rho\ge 1, \rho=\Theta(1)$\\

\cite{ABF+25} & $O(\frac{n}{\rho} \log n)$ & $O(\frac{n}{\rho} \log n)$  & $O(\rho^2 \log n)$ & $1 \le \rho < n/2$\\

Theorem \ref{theorem:main} & $O(\frac{n}{\rho}\log \rho)$ & $O(\frac{n}{\rho}\log \rho \cdot \log n)$ & $2\rho \lg^2 \rho + O(\log n)$ & $2 \le \rho \le \sqrt{n}$ \\

($\rho = \lfloor \sqrt{n} \rfloor$) & $O(\sqrt{n}\log n)$ & $O(\sqrt{n}\log^2 n)$& $\sqrt{n} \lg^2 n + O(\log n)$& - \\
($\rho = \frac{\log n}{\log^2 \log n}$)& $O(\frac{n\log^3 \log n}{\log n})=o(n)$ & $O(n\log^3 \log n)$
&
$O(\log n)$
& - \\

\hline 
trivial & $\Omega(\log n)$& $\Omega(\log n)$  & any &-\\
\hline
\end{tabular} 
\end{table}

After the seminal work of~\cite{BCC+21}, the SS-LE problem with exact knowledge of the population size $n$ has recently attracted renewed attention. In particular, four papers appeared in 2025: Berenbrink, Elsässer, Götte, Hintze, and Kaaser~\cite{BEG+25}; Gąsieniec, Grodzicki, and Stachowiak~\cite{GGS25}; Araya and Sudo~\cite{AS26}\footnote{The conference version~\cite{AS25} appeared in 2025.}; and Austin, Berenbrink, Friedetzky, Götte, and Hintze~\cite{ABF+25}.

The first two studies, \cite{BEG+25} and \cite{GGS25}, advance silent SS-LE. 
Since any SS-LE protocol requires at least $n$ states, a central question is how many additional states reduce the $O(n^2)$ stabilization time of the space-optimal protocol of \cite{CIW12}. The authors of \cite{BEG+25} showed that $n+O(\log^2 n)$ states achieve optimal $O(n\log n)$ time w.h.p., and \cite{GGS25} reduced the overhead to $O(\log n)$ states while preserving optimal w.h.p.\ time. 
They further proved that even one additional state yields $O(n^{7/4}\log^2 n)$ time w.h.p., still maintaining silence.

For non-silent SS-LE, \cite{BCC+21} showed that dropping silence enables sublinear stabilization time, but all their such protocols require super-exponentially many states (i.e., $\omega(n)$ bits). If Conjecture~\ref{conjecture:ss-rk} holds, then fast SS-RK (and hence SS-LE) necessarily involves detecting rank collisions (i.e., two or more agents sharing the same rank). In their words, ``the core difficulty'' of fast SS-RK lies in collision detection, and their non-silent protocols use super-exponentially many states for this task.  Motivated by this, \cite{BCC+21} posed the following open problems:
\begin{enumerate}
  \item Does there exist an SS-LE protocol with sublinear stabilization time using only $o(n)$ bits (i.e., sub-exponentially many states)?
  \item Does there exist a protocol that detects collisions in sublinear time using only $o(n)$ bits, even in the non-self-stabilizing setting (where agents start with adversarially assigned ranks but otherwise share the designated initial state)?
\end{enumerate}

The second question was answered affirmatively in~\cite{AS26}, which showed that $O(n\poly(\log n))$ states suffice to detect collisions in polynomial time. 
In particular, in the non-self-stabilizing setting, collisions can be detected in $O(\sqrt{n\log n})$ expected time.
However, this protocol does not directly resolve the core difficulty of fast SS-RK. 
It runs a fast primary protocol and a slow backup protocol in parallel: the primary stabilizes in $O(\sqrt{n}\,\log^{3/2} n)$ time w.h.p.\ but eventually terminates to avoid false positives, whereas the backup never terminates and requires $\Theta(n)$ expected time. 
Although together they detect collisions in $O(\sqrt{n\log n})$ expected time, applying this approach to SS-RK requires repeatedly restarting the primary protocol. 
Without this reset, execution may start in a configuration where the primary protocol has already terminated; in that case, only the backup protocol can detect collisions, incurring an inevitable $\Theta(n)$ delay.
However, repeating the cycle of the primary protocol is also problematic: since interactions are chosen uniformly at random (u.a.r.) at each time step, there remains a nonzero (albeit arbitrarily small) probability that a global reset fails and some agents retain outdated information, leading to false positives and repeated rank reassignment, thereby violating stability.

Recently, \cite{ABF+25} presented a dramatically improved time--space tradeoff for SS-LE, affirmatively answering the first open question. They proposed a protocol parameterized by a positive integer $\rho$ ($1 \le \rho < n/2$) that solves SS-RK (and thus SS-LE) in $O(\sfrac{n}{\rho}\cdot\log n)$ time w.h.p.\ using $O(\rho^2\log n)$ bits.
Choosing $\rho = \omega(\log n)\cap o(\sqrt{n/\log n})$ yields a sublinear-time, sub-exponential-state protocol:
for example, $\rho=\Theta(\sqrt{n}/\log n)$ gives $O(\sqrt{n}\log^2 n)$ time w.h.p.\ with $2^{O(n/\log n)}=2^{o(n)}$ states, while $\rho=\Theta(\log^2 n)$ gives $O(n/\log n)$ time w.h.p.\ with $2^{O(\log^5 n)}$ states.
Allowing super-exponentially many states---but still far fewer than in \cite{BCC+21}---yields a time-optimal solution: setting $\rho=\Theta(n)$ achieves $\Theta(\log n)$ time w.h.p.\ using $2^{O(n^2\log n)}$ states.
Their collision-detection module also resolves the second question by \cite{BCC+21} independently of~\cite{AS26}. Although this module uses significantly more states ($2^{O(\rho^2\log n)}$ vs.\ $O(n\cdot\mathrm{poly}(\log n))$), it never terminates and thus does not require repeated invocation, enabling the first sublinear-time, sub-exponential-state SS-RK protocol.

While the majority of work on population protocols studies the setting where every pair of agents can interact, quite a few studies deal with a more general setting: interactions may occur between restricted pairs of agents, introducing an \emph{interaction graph} $G=(V,E)$, where $V$ represents the set of agents and $E$ the set of interactable pairs.
SS-LE and its weaker variant have also been studied in rings~\cite{FJ06,CC19,YSM21,YSO+23}, regular graphs~\cite{CC20}, and general graphs~\cite{FJ06,SOK+18,SOK+20det,SSN+21,KESI24}.

\subsection{Our Contribution}
\label{sec:contribution}
This paper affirmatively resolves the first open problem of \cite{BCC+21}, independently of and nearly concurrently with \cite{ABF+25}.\footnote{
The first version of this paper was posted on arXiv in May 2025, the same month that~\cite{ABF+25} appeared on arXiv.
} Specifically, we prove the following theorem. (We write $\lg x$ for $\log_2 x$ throughout this paper.)

\begin{theorem}[Main Theorem]
\label{theorem:main}
For any positive integer $\rho$ with $2 \le \rho \le \sqrt{n}$, there exists a protocol that solves SS-RK (and hence SS-LE) in $O(\goalT)$ expected time using $2^{2\rho\lg^2 \rho + O(\log n)}$ states
(equivalently, $2\rho\lg^2 \rho + O(\log n)$ bits).
\end{theorem}
Choosing any $\rho=\omega(1)$ yields a sublinear-time, sub-exponential-state protocol. For example, $\rho=\lfloor \sqrt{n}\rfloor$ achieves $O(\sqrt{n}\log n)$ expected time with $2^{\sqrt{n} \lg^2 n}\cdot\poly(n)$ states, while $\rho=\Theta\left(\sfrac{\log n}{\log^2\log n}\right)$ achieves $O\left(n \cdot \sfrac{\log^3\log n}{\log n}\right)=o(n)$ expected time with only $\poly(n)$ states (see Table \ref{table:non-silent}).

The proposed protocol and that of \cite{ABF+25} exhibit complementary strengths and trade-offs.
The proposed protocol uses significantly fewer states than the protocol of \cite{ABF+25} requires to achieve any expected stabilization time above $\Theta(\sqrt{n}\log n)$. However, the restriction $\rho\le\sqrt{n}$ prevents us from attaining stabilization time $o(\sqrt{n}\log n)$. Conversely, the protocol of \cite{ABF+25} can exceed this bound with super-exponentially many states; in particular it can achieve an optimal $O(\log n)$-time solution with $2^{O(n^2\log n)}$ states.

The second example mentioned above (i.e., $\rho = \Theta\left(\sfrac{\log n}{\log^2\log n}\right)$) resolves the first open question of~\cite{BCC+21} in stronger form, upgrading the state bound from sub-exponential to \emph{polynomial}:
\begin{corollary}
\label{col:strong}
There exists a protocol that solves SS-RK (and hence SS-LE) in sublinear expected time using only \emph{polynomially} many states.
\end{corollary}
Note that the protocol of~\cite{ABF+25} does not achieve this: to stabilize in sublinear time it requires $\rho=\omega(\log n)$, which entails $\omega(\log^2 n)$ bits (i.e., super-polynomial states). One might observe that \cite{ABF+25} only claim an $O(\sfrac{n}{\rho}\cdot\log n)$ bound for w.h.p.\ stabilization time and hope that a refined analysis or slight modification could reduce their expected time to $O(\sfrac{n}{\rho}\cdot\log \rho)$ or even $O(\sfrac{n}{\rho})$.
However, their SS-RK protocol employs a non-self-stabilizing ranking subprotocol, $\assignranks_{\rho}$ (for $1 \le \rho < n/2$), which inherently requires $\Omega(\sfrac{n}{\rho}\cdot\log n)$ expected time to assign distinct ranks to all $n$ agents while using only $2^{O(\rho \log n)}$ states.

\begin{remark}[Proof in Appendix~\ref{sec:proofs}]
\label{remark:lower_bound_abf}
The ranking protocol $\assignranks_{\rho}$ of~\cite{ABF+25} requires $\Omega(\sfrac{n}{\rho}\cdot\log n)$ expected time to assign distinct ranks to all $n$ agents.
\end{remark}

It is worth mentioning that the proposed SS-LE protocol employs a \emph{loosely-stabilizing leader election} (LS-LE) protocol as a subroutine. LS-LE~\cite{SNY+12, SOK+20, SEIM21} is a weaker variant of SS-LE: from any configuration, the population quickly reaches a \emph{safe} configuration from which exactly one leader persists for an arbitrarily long (but finite) expected time. The time-optimal LS-LE protocol of Sudo, Eguchi, Izumi, and Masuzawa~\cite{SEIM21}, parameterized by $\tau\ge1$, reaches a safe configuration in $O(\tau \log n)$ expected time and thereafter maintains a unique leader for $\Omega(n^{\tau})$ expected time, using only $O(\tau\log n)$ states. Since its time and space complexities are negligible, we can assume a unique leader \emph{almost always} exists when designing our SS-LE protocol. To the best of our knowledge, this is the first work where an LS-LE protocol aids the design of an SS-LE protocol---an intriguing approach in which a stable leader is elected with the help of an unstable leader. 

Most modern population protocols employ the epidemic protocol~\cite{AAE08} to disseminate information from a single agent to the entire population, whereas our protocol uses a slightly slower propagation mechanism to avoid false-positive collision detection. To analyze its stabilization time, we introduce a new variant of the epidemic protocol, called the \emph{labeled epidemic}. While the original epidemic spreads information to all agents within $O(\log n)$ time w.h.p.\ and in expectation, we show that this variant spreads information to a constant fraction of the population (i.e., $\Theta(n)$ agents) in $O(\log^2 n)$ expected time, which may be of independent interest in the study of population protocols.

\subsection{Notation}
Throughout this paper, we write
$[x,y] = \{ i \in \mathbb{Z} : x \le i \le y\}$
for any real numbers $x$ and $y$.
For any set of configurations $\calC$, we say the population \emph{enters} $\calC$ when it reaches some configuration in $\calC$. 

We often use \emph{variables} to define protocols.
In this context, an agent’s state is determined by the values of all protocol variables. For any state $q\in Q$ and variable $\var$, we denote by $q.\var$ the value of $\var$ in state $q$. For any agent $a\in A$, we denote by $a.\var$ the value of $\var$ in $a$’s current state.

\section{Protocol Overview}
\label{sec:overview}
We prove our main theorem (Theorem~\ref{theorem:main}) by presenting the SS-RK protocol $\ssrk(\rho)$. Given any integer $\rho \in [2,\sqrt{n}]$, this protocol solves SS-RK
within $O(\goalI)$ interactions (i.e., $O(\goalT)$ time) in expectation using $2^{2\rho\lg^2 \rho + O(\log n)}$ states.
The pseudocode for SS-RK is presented in Algorithms~\ref{al:ccd}, \ref{al:main}, \ref{al:ranking}, and~\ref{al:findtarget} on pages~\pageref{al:ccd}, \pageref{al:main}, \pageref{al:ranking},
and~\pageref{al:findtarget}, respectively.

The protocol $\ssrk(\rho)$ comprises three sub-protocols: $\findtarget$, $\ccd(r,\rho)$, and $\ranking$. 
These components serve the following roles:\\

\begin{tabular}{|>{\centering\arraybackslash}m{2cm}|m{12cm}|}
\hline
$\findtarget$ & if some rank $r\in[1,n]$ is held by two or more agents, selects one such $r$ as the \emph{target rank} with probability $1-o(1)$. \\
\hline
$\ccd(r,\rho)$ & if the target rank $r$ is held by two or more agents, detects the collision with probability $\Omega(1)$. \\
\hline
$\ranking$ & upon collision detection, reassigns unique ranks in $[1,n]$ to all agents w.h.p. \\ 
\hline
\end{tabular}
\vspace{0.45cm}

\noindent
We execute these components sequentially in cycles. To ensure the above success probabilities, the sub-protocols $\findtarget$, $\ccd(r,\rho)$, and $\ranking$ require $O(n^{3/2}\log n)$, $O(\goalI)$, and $O(n^{3/2}\log n)$ interactions, respectively; the second bound dominates for all $\rho\in[2,\sqrt{n}]$. By the union bound, each cycle succeeds with probability $1 - (o(1) + (1 - \Omega(1)) + o(1)) = \Omega(1)$. Therefore, the expected number of cycles to success is $O(1)$, yielding $O(\goalI)$ expected interactions overall.

The core novelty of $\ssrk(\rho)$ lies in $\ccd(r,\rho)$. As noted in Section~\ref{sec:contribution}, repeating a collision-detection protocol with periodic resets carries a small but nonzero probability of reset failure, leaving outdated information in some agents. Such residual data causes false-positive detections and triggers rank reassignment, preventing stabilization. Consequently, the polynomial-time, polynomial-state protocol of~\cite{AS26} does not immediately yield a polynomial-time, polynomial-state SS-RK protocol. The authors of \cite{ABF+25} resolve this issue by avoiding periodic resets: they reinitialize their collision-detection module only when a collision is detected. In contrast, our component $\ccd(r,\rho)$ tolerates outdated information: once all agents have initialized their variables at least once, $\ccd(r,\rho)$ incurs no false positives, even if some agents retain outdated information from cycles long past.
Therefore, once a single complete cycle of the three components successfully assigns unique ranks, $\ccd(r,\rho)$ produces no false positives in subsequent cycles, ensuring stabilization.

The tolerance to outdated data in $\ccd(r,\rho)$ enables us to incorporate the time- and space-efficient collision-detection protocol of~\cite{AS26}. 
We obtain $\findtarget$ by modifying their protocol to output a target rank $r$ shared by at least two agents, if one exists. 
Restricting detection to a single rank $r$ significantly reduces the state complexity of $\ccd(r,\rho)$: agents only need to track collisions on $r$, resulting in fewer states than the protocol of~\cite{ABF+25}.

\paragraph*{Paper Organization}
The rest of this paper is organized as follows.
We begin with the component containing the core novelty, $\ccd(r,\rho)$. 
Specifically, Section~\ref{sec:epidemic} introduces the labeled epidemic and analyzes its infection speed; 
Section~\ref{sec:ccd} then presents $\ccd(r,\rho)$,
bounds its stabilization time using the labeled epidemic,
and proves its safety.
Thereafter, Section~\ref{sec:tools} introduces the existing tools used in $\ssrk(\rho)$, 
Section~\ref{sec:cycle} describes how to maintain cyclic execution of the three components, 
and Sections~\ref{sec:ranking} and~\ref{sec:findtarget} present $\ranking$ and $\findtarget$, respectively. 
Finally, Section~\ref{sec:main_theorem} proves Theorem~\ref{theorem:main} by combining the preceding lemmas, 
and Section~\ref{sec:conclusion} concludes the paper and poses open questions.

\section{Labeled Epidemic}
\label{sec:epidemic}
The \emph{epidemic} protocol~\cite{AAE08}, widely used in population protocols, propagates the maximum value of a given variable $\var$ to all agents. It consists of the single transition rule
$$
\res.\var \gets \max(\ini.\var, \res.\var),
$$
where $\ini$ and $\res$ denote the initiator and responder, respectively.
Here, we call an agent with the maximum $\var$ value \emph{red}, and any other agent \emph{blue}.
By this notation, the following lemma holds:
\begin{lemma}[\cite{AAE08}]
\label{lemma:epidemic}
Under the epidemic, all agents become red in $O(n \log n)$ interactions w.h.p.
\end{lemma}

We now introduce the \emph{labeled epidemic}.
\begin{definition}[Labeled Epidemic]
Suppose that all agents have distinct labels in $[1,n]$.
In each interaction, a red initiator turns the responder red
if and only if the initiator's label is smaller.
\end{definition}
\noindent
Note that we do not assume unique agent labels in $\ssrk(\rho)$; such an assumption would trivialize SS-RK.
Nevertheless, the labeled epidemic is useful for analyzing the stabilization time of $\ccd(r,\rho)$, as shown in the next section.

\begin{lemma}
\label{lemma:labeled_epidemic}
Under the labeled epidemic, if the agent labeled $1$ is initially red,
the number of red agents reaches $\lceil n/4 \rceil$ in $O(n \log^2 n)$ interactions in expectation.
\end{lemma}

\begin{proof}
Let $\psi = \lfloor n/(2 \lg n) \rfloor$ and $\kappa = \lfloor \lg n - \lg \lg n -1 \rfloor$. For each $i \in [1,\kappa]$, define $S_i=[(i-1)\psi+1,i\psi]$ and let $A_i$ be the set of agents whose labels lie in $S_i$.

Fix $i \in [1,\kappa-1]$.
Suppose that $2^{i-1}$ agents in $A_i$ are red. We show that the number of red agents in $A_{i+1}$ reaches at least $2^i$ within the next $O(n\log n)$ expected interactions.
As long as fewer than $2^{i}$ agents in $A_{i+1}$ are red, at least $\psi - 2^{i} \ge \psi-2^{\kappa-1} = \Omega(n/\log n)$ agents in $A_{i+1}$ remain blue.
Hence, in each interaction, the number of red agents in $A_{i+1}$ increases by one with probability at least
$$
\frac{2^{i-1} \cdot \Omega(n/\log n)}{n(n-1)} = \Omega\left(\frac{2^i}{n \log n}\right),
$$
because at least $2^{i-1}$ agents are red in $A_i$.
Thus, $O(n \log n/2^i)\cdot 2^{i} = O(n \log n)$ expected interactions suffice to turn $2^{i}$ agents in $A_{i+1}$ red.

Since the agent labeled $1$ is initially red, by induction $2^\kappa = \Theta(n/\log n)$ agents in $A_\kappa$ become red within $O(\kappa n \log n) = O(n \log^2 n)$ expected interactions.

Before the total number of red agents reaches $\lceil n/4 \rceil$, there are always at least $\lfloor n/4 \rfloor$ blue agents whose labels are at least $n/2$. Each such blue agent becomes red upon interacting with one of the $2^\kappa = \Theta(n/\log n)$ red agents in $A_\kappa$, whose labels are all less than $\kappa \psi < n/2$. Consequently, the number of red agents reaches $\lceil n/4 \rceil$ within an additional
$
\lceil n/4 \rceil \cdot \frac{n^2}{\Theta(n/\log n)\cdot \lfloor n/4 \rfloor} = O(n \log n)
$
expected interactions.

Summing the two phases yields $O(n \log^2 n)$ expected interactions in total.
\end{proof}

\begin{algorithm}[t]
\caption{
$\ccd(r,\rho)$ at an interaction where initiator $\ini$ and responder $\res$ meet. 
Whenever this sub-protocol is invoked, $r = \ini.\target = \res.\target$.
} 
\label{al:ccd}
\Variables{
$\vlist \in \{S \subseteq \rname: |S| \le \maxnum\}$, 
$\susp \in \{0,1\}$
}
\Initially{
$a.\detected = 0$, $a.\vlist = \emptyset$, $a.\susp = 0$
}
\uIf{$(\ini,\res) \in A_K \times A_R\, \land\, \res.\rank \notin \ini.\vlist \, \land \, |\ini.\vlist| <  \maxnum$}{
$\res.\vlist \gets \ini.\vlist \gets \ini.\vlist \cup \{\res.\rank\}$\;
}
\uElseIf{$(\ini,\res) \in A_K  \times A_V\, \land\, \res.\vlist \nsubseteq \ini.\vlist$}{
$\res.\susp \gets 1$\;
}
\ElseIf{$(\ini,\res) \in A_V \times A_V\, \land\, \ini.\vlist \subseteq \res.\vlist$}{
$\res.\susp \gets \max(\ini.\susp,\res.\susp)$\;
}
$
\res.\detected \gets 
\begin{cases}
1 & \tif \ini.\rank = \res.\rank \\
\max(\ini.\susp,\res.\detected) & \telseif 
(\ini,\res) \in A_V \times A_K \land \ini.\vlist \subseteq \res.\vlist\\
\max(\ini.\detected,\res.\detected) & \totherwise
\end{cases}
$\;
\end{algorithm}

\section{Careful Collision Detection}
\label{sec:ccd}
Given a target rank $r$, the subprotocol $\ccd(r,\rho)$, shown in Algorithm~\ref{al:ccd},
detects whether at least two agents share rank $r$
and stores the result in each agent's variable $\detected \in \{0,1\}$,
where $1$ indicates detection.
Once any agent sets its $\detected$ flag to $1$, this information propagates to the entire population via the epidemic protocol in $O(n \log n)$ interactions w.h.p.~(line~7).
Therefore, it suffices for some agent to detect a collision, if one exists.
If agents collide on some rank $r' \neq r$, $\ccd(r,\rho)$ may report a collision even when only one agent holds rank $r$.
Such misdetections are harmless, since any detected collision should trigger a rank reassignment.

As we will describe in Section~\ref{sec:cycle}, each agent maintains variables 
$\mode \in \{\bot, \modeF, \modeD, \modeR\}$ and $\target \in [1,n]$.
The main function invokes $\ccd(r,\rho)$ only when two agents with $\mode = \modeD$
and the same $\target$ interact, where the target rank $r$ is given by their common $\target$ value.
In $\ccd(r,\rho)$, each agent maintains a set variable $\vlist \subseteq \rname$ with $|\vlist| \le \maxnum$,
where $\rname = [1,\mname]$, and a Boolean variable $\susp \in \{0,1\}$.
So, this subprotocol only uses $2^{2\rho \log^2 \rho+O(1)}$ states.
When an agent $a$ switches to mode $\modeD$ from another mode, it initializes
$a.\detected = 0$, $a.\vlist = \emptyset$, and $a.\susp = 0$.

We classify agents in $\ad = \{a \in A \mid a.\mode = \modeD\}$ into four roles as follows:
\begin{align*}
&A_K = \{a \in \ad \mid a.\rank = a.\target \}, \quad
A_V = \{a \in \ad \mid a.\rank \in a.\vlist\},\\
&A_R = \{a \in \ad \setminus A_K \mid a.\rank \in \rname \land a.\vlist = \emptyset \land a.\susp = 0\}, \quad
A_C = \ad \setminus (A_K \cup A_V \cup A_R).
\end{align*}
Agents in $A_K$, $A_V$, $A_R$, and $A_C$ are called \emph{kings}, \emph{vassals}, \emph{ronins} (i.e., a masterless warrior),
and \emph{commoners}, respectively.
This subprotocol ignores commoners except for direct collision detection:
whenever two agents of the same rank interact,
the responder immediately sets $\detected \gets 1$ (line~7).

Direct collision detection simplifies the analysis:
without loss of generality, we may restrict attention to initial configurations
with few colliding pairs, particularly those with few kings and many ronins,
as implied by the following simple observation.

\begin{observation}[Generalization of Lemma~3 of \cite{AS26}]
\label{obs:many_ronins}
If there are $x\ge1$ colliding pairs (i.e., unordered pairs of agents sharing the same rank), the direct collision detection protocol detects a collision within $O(n^2/x)$ interactions both in expectation and with probability $\Omega(1)$.
\end{observation}
\begin{proof}
Each interaction involves a colliding pair with probability at least $x/\binom{n}{2} = \Omega(x/n^2)$, so $O(n^2/x)$ expected interactions suffice for collision detection.
The constant‐probability bound then follows by Markov’s inequality.
\end{proof}

The goal of a king is to detect the presence of another king. 
To this end, each king adopts up to $\maxnum$ vassals and records their ranks in $\vlist$: 
when a king $a_K$ with $|a_K.\vlist| < \maxnum$ meets a ronin $a_R$ such that $a_R.\rank \notin a_K.\vlist$, it adds $a_R.\rank$ to $a_K.\vlist$ and copies the updated list to $a_R.\vlist$, thereby making $a_R$ a vassal (lines~1--2). 
After $a_K$ adopts $\maxnum$ vassals $a_1,\dots,a_{\maxnum}$ in this order, we have
$$
a_K.\vlist = \{a_i.\rank \mid i \in [1,\maxnum]\},
\quad
\forall j \in [1,\maxnum]: a_j.\vlist = \{a_i.\rank \mid i \in [1,j]\}.
$$
Hence, if one of those vassals, say $a_i$, encounters a king $a'_K$ with $a_i.\vlist \nsubseteq a'_K.\vlist$, this implies that at least two kings exist. In this case, the vassal $a_i$ sets its $\susp$ flag to $1$ (lines~3--4), and this flag propagates to vassals via a variant of the epidemic protocol. Specifically, if two vassals $a$ and $b$ meet where $a.\vlist \subseteq b.\vlist$ and $a.\susp = 1$, then $b$ also raises its $\susp$ flag (lines~5--6). This propagation can be regarded as the labeled epidemic that involves
only $\maxnum-i+1$ agents $a_i,a_{i+1},\dots,a_{\maxnum}$.
When $a_K$ later meets one of those vassals with $\susp = 1$,
it sets its $\detected$ flag to $1$ (line~7).

A vassal $a_V$ may raise its $\susp$ flag even when there is exactly one king with $\rank = a_V.\target$ if it is outdated—i.e., its king has already started the next cycle and reset its vassal list, so the vassal’s $\vlist$ is no longer a subset of the king’s $\vlist$. 
(We formally define outdated vassals in Section \ref{sec:safety}.)
Although the phase clock synchronizes the population w.h.p., making this event extremely rare, we must still handle it to solve SS-RK (whereas loosely-stabilizing ranking could safely ignore it). Our variant of the epidemic protocol addresses this issue: even if an outdated vassal raises its $\susp$ flag, the information propagates only among outdated vassals, and hence the unique king never raises its $\detected$ flag. This \emph{safety property} (i.e., the absence of false positives) is formalized as Lemma \ref{lemma:closure} in Section \ref{sec:safety}. 

\subsection{Liveness}
To guarantee the absence of false positives, we restrict the propagation of $\susp$.
As noted above, this can be viewed as a labeled epidemic, which allows us to give the following lemma via Lemma~\ref{lemma:labeled_epidemic}.

\begin{lemma}
\label{lemma:detect}
Let $r$ be an integer in $[1,n]$. 
Let $C$ be a configuration in which all agents have $\susp = 0$, $\vlist = \emptyset$, $\target = r$, and $\mode = \modeD$, and at least two agents have $\rank = r$. 
Starting from $C$, $\ccd(r,\rho)$ reaches a configuration in which every agent has $\detected = 1$ within $O(\goalI)$ interactions with probability $\Omega(1)$.
\end{lemma}

\begin{proof}
Fix any king $a_K$ in $C$. It suffices to show that $a_K$ raises its $\detected$ flag within $O(\goalI)$ interactions with probability $\Omega(1)$; by Lemma~\ref{lemma:epidemic}, this flag then propagates to all agents within $O(n\log n)$ interactions w.h.p. We may assume w.l.o.g.\ that $\rho=\omega(1)$, since if $\rho=O(1)$, direct collision detection alone raises $\detected$ flag within $O(n^2)$ interactions with probability $\Omega(1)$, as desired.

Let $x_C$, $n_K$, and $n_R$ be the numbers of colliding pairs (i.e., unordered pairs of agents sharing the same rank), kings, and ronins, respectively, in configuration $C$. By Observation~\ref{obs:many_ronins}, we may assume w.l.o.g.\ that $x_C = o(\rho/\log \rho)$. Since $x_C \ge \binom{n_K}{2}$, this implies $n_K = o(\sqrt{\rho/\log \rho})$.
Because $C\in\cinit(\modeD,r)$ contains no vassals, only ronins (and kings if $r \le \mname$) occupy ranks in $\rname=[1,\mname]$. If $n_R+n_K<|\rname|$, the pigeonhole principle gives $x_C \ge |\rname|-n_K-n_R$, thus $n_R \ge |\rname| - n_K - x_C$;
Otherwise, $n_R > |\rname|-n_K$.
Hence $n_R \ge |\rname|-o(\sqrt{\rho/ \log \rho}) - o(\rho/\log \rho)=\Omega(\mname)$.

Since each of the $n_K$ kings adopts at most $\maxnum$ vassals, there are always at least 
$n_R - n_K \cdot \maxnum = \Omega(\mname)$ ronins. 
Thus, in each step, $a_K$ meets a ronin with $\rank \notin a_K.\vlist$ with probability 
$(\Omega(\mname) - x_C)/n(n-1) = \Omega(\rho^2/n^2)$, 
since at most $x_C$ ronins have ranks already contained in $a_K.\vlist$ 
(each such ronin yields a colliding pair). 
Hence, $a_K$ adopts $\maxnum$ vassals within 
$\maxnum \cdot O(n^2/\mname) = O(\goalI)$ expected interactions.

Let $a_1,\dots,a_{\maxnum}$ be the vassals adopted by $a_K$ in this order. 
At any moment,
at most $2x_C = o(\rho/\log \rho)$ agents share a rank with another agent. 
Therefore, $a_1$ has a unique rank with probability at least 
$(|A'| - 2x_C)/|A'| = 1 - o((\rho \log \rho)/\rho^2) = 1-o(1)$, 
where $A'$ is the set of ronins immediately before $a_K$ adopts $a_1$. 
Hence, w.l.o.g., we assume that $a_1$ has a unique rank.

Let $a'_K$ be any king in $A_K \setminus \{a_K\}$. 
After $a_K$ adopts $\maxnum$ vassals, 
$a'_K$ meets one of the first $\lfloor \maxnum/2 \rfloor$ vassals 
(i.e., $a_1,\dots,a_{\lfloor \maxnum/2 \rfloor}$) 
within $O(n^2/\maxnum) = O(\goalI)$ expected interactions. 
Let $a_i$ denote this vassal. 
Since $a_1$ has a unique rank and $a_i.\vlist$ contains $a_1.\rank$, 
we have $a_i.\vlist \nsubseteq a'_K.\vlist$, 
and hence $a_i$ raises its $\susp$ flag.
For any $j,k$ with $i \le j \le k \le \maxnum$, 
$a_j.\vlist \subseteq a_k.\vlist$. 
Thus, the propagation of $\susp$ among 
$a_i, a_{i+1}, \dots, a_{\maxnum}$ 
can be viewed as a labeled epidemic with $a_i$ initially red. 
In each step, a pair of these vassals meets with probability 
$\Omega(\rho^2 \log^2 \rho / n^2)$. 
By Lemma~\ref{lemma:labeled_epidemic}, 
at least $\maxnum/8$ of them raise $\susp$ flag within
$$
O\left(\frac{n^2}{\rho^2 \log^2 \rho}\right)
\cdot O((\rho \log \rho)\log^2(\rho \log \rho))
= O(\goalI)
$$
expected interactions.
Finally, $a_K$ meets one of these vassals and raises $\detected$ flag
within $O(n^2/(\rho \log \rho)) = O(\goalI)$ expected interactions.

Therefore, the total expected number of interactions is $O(\goalI)$. 
By Markov's inequality, this implies a constant success probability.
\end{proof}

\subsection{Safety}
\label{sec:safety}
In this subsection, we prove the safety property (i.e., the absence of false positives), stated as Lemma~\ref{lemma:closure}. 
To state the lemma formally, we first define outdated vassals and safe configurations.

\begin{definition}[Outdated Vassals]
Define a function $f_K: A \to A \cup \{\bot\}$ as follows.
For any agent $a \in \ad$, if there exists \emph{exactly one} king $a_K \in A_K$ such that $a.\target = a_K.\target = a_K.\rank$, then $f_K(a) = a_K$; otherwise, $f_K(a) = \bot$.
A vassal $a_V \in A_V$ is called an \emph{outdated vassal}
if $f_K(a_V) = \bot$ or $a_V.\vlist \nsubseteq f_K(a_V).\vlist$.
We denote by $A_{OV}$ the set of outdated vassals.
\end{definition}

\begin{definition}[Safe Configurations]
\label{def:safe}
A configuration is \emph{safe} if the following conditions hold:
(C1) all agents have distinct ranks;
(C2) no agent is in mode $\modeR$;
(C3) all agents have $\detected = 0$; and
(C4) all agents in $A_V \setminus A_{OV}$ have $\susp = 0$.
\end{definition}

\begin{remark}
\label{remark:spec}
As stated in Section~\ref{sec:overview}, collision detection and rank reassignment are handled by the subprotocols $\ccd(r,\rho)$ and $\ranking$, respectively, and $\ranking$ is executed only when a rank collision is detected. 
Therefore, in $\ssrk(\rho)$, no agent raises the $\detected$ flag outside $\ccd(r,\rho)$, 
no agent changes its mode to $\modeR$ (the mode for $\ranking$) unless $\detected = 1$, 
and no agent changes its $\rank$ unless its mode is $\modeR$. 
(A complete specification of $\ssrk(\rho)$ is given in Section~\ref{sec:cycle}.)
\end{remark}

\begin{lemma}
\label{lemma:closure}
The set of safe configurations is closed under $\ssrk(\rho)$.
That is, if a safe configuration $C$ changes to $C'$ by a single interaction under $\ssrk(\rho)$, then $C'$ is also safe.
\end{lemma}

\begin{proof}
By Remark~\ref{remark:spec}, (C1) and (C2) are preserved in $C'$. 
Condition (C3) is preserved, since in $\ccd(r,\rho)$ only a king $a_K$ may raise the $\detected$ flag, which would require a non-outdated vassal with $\susp = 1$ and $\target = a_K.\rank$, contradicting the safety of $C$. 
Condition (C4) involves only the variables $\susp$, $\vlist$, and $\mode$, and the first two are updated only by $\ccd(r,\rho)$.
It therefore suffices to verify that (C4) is preserved in the following cases:
(i) an agent in $A \setminus \ad$ switches to mode $\modeD$,
(ii) an agent in $\ad$ leaves mode $\modeD$, or
(iii) $\ccd(r,\rho)$ is executed.

It remains to show that no non-outdated vassal with $\susp = 1$ is created in cases (i)--(iii). 
A new vassal arises only when a ronin is adopted by a king. (Event (i) never creates a vassal.) 
Since every ronin has $\susp = 0$ by definition, any newly created vassal also has $\susp = 0$. 
Hence only two possibilities remain: 
a non-outdated vassal raises its $\susp$ flag, or 
an outdated vassal with $\susp = 1$ becomes non-outdated.

The former case never occurs. 
By definition, a non-outdated vassal never raises its $\susp$ flag upon meeting a king. 
Hence it can only happen when a non-outdated vassal $a$ meets a vassal $b$ with 
$b.\susp = 1$ and $b.\vlist \subseteq a.\vlist$. 
Since $C$ is safe, such a $b$ must be outdated. 
As $a.\target = b.\target$, we have $f_K(a) = f_K(b)$. 
If $f_K(a) = \bot$, then $a$ is outdated in $C$, a contradiction. 
If $f_K(a) \neq \bot$, then $b.\vlist \nsubseteq f_K(a).\vlist$ while $a.\vlist \subseteq f_K(a).\vlist$, contradicting $b.\vlist \subseteq a.\vlist$.

The latter case also never occurs. 
An outdated vassal $a_V$ can become non-outdated after a single interaction,
only if either a new king is created or $f_K(a_V)$ adds $a_V.\rank$ to its $\vlist$. 
A newly created king initializes $\vlist$ to $\emptyset$, so it cannot make $a_V$ non-outdated. 
If $f_K(a_V)$ adds $a_V.\rank$, it must meet a ronin of rank $a_V.\rank$, but no such ronin exists since ranks are distinct in $C$.
\end{proof}

As we will see in Section~\ref{sec:cycle}, $\ssrk(\rho)$ synchronizes the population using the phase clock protocol, and all agents synchronously maintain the cycle of their modes,
i.e., $\bot \to \modeF \to \modeD \to (\modeR \to) \bot \to \modeF \to \modeD \to \cdots$, where they enter $\modeR$ only when they raise the $\detected$ flag.
Therefore, once unique ranks are assigned by $\ranking$ in mode $\modeR$ and all agents subsequently switch to $\bot$, the population is in a safe configuration.
Since each cycle finishes in $O(\goalI)$ interactions w.h.p., the population reaches a safe configuration in $O(\goalI)$ expected interactions.
By Lemma~\ref{lemma:closure}, no false positive occurs thereafter.

\section{Existing Tools}
\label{sec:tools}

\subsection{Loosely-stabilizing Leader Election}
Several components of $\ssrk(\rho)$ assume a unique leader exists; to satisfy this, we run in parallel the LS-LE protocol of \cite{SEIM21}, denoted by $\lsle$ here (line~8 in Algorithm \ref{al:main}).  

The following lemma allows us to assume, throughout the rest of this paper, except for the proof of Theorem~\ref{theorem:main} in Section~\ref{sec:main_theorem}, that there is exactly one agent $a_L$ with $a_L.\leader = 1$, and every other agent $a \in A \setminus \{a_L\}$ has $a.\leader = 0$, where $\leader \in \{0,1\}$ is the leader-indicator variable. This assumption will be removed in the proof of the main theorem using this lemma.

\begin{lemma}[immediate consequence of \cite{SEIM21}]
\label{lemma:lsle}
For any $T=\poly(n)$, there exists a constant $\tau=\Theta(1)$ and a set $\csafe$ of configurations such that:
  (i) starting from any configuration, $\lsle(\tau)$ enters $\csafe$ within $O(\tau n\log n)$ expected interactions, and
  (ii) once in $\csafe$, it maintains a unique leader for at least $T$ interactions w.h.p. 
\end{lemma}
\noindent
For completeness, we include the proof of the above lemma in Appendix~\ref{sec:proofs}.

\begin{note}
Readers who are already familiar with population protocols may skip Section~\ref{sec:phase_clock}.
\end{note}

\subsection{Phase Clock with a Single Leader}
\label{sec:phase_clock}
Many variants of the phase clock have been used to loosely synchronize the population in protocol design.
We use the simplest variant: the single-leader phase clock of~\cite{AAE08}. This protocol assumes a fixed unique leader.
Each agent maintains a single variable $\clock \in [0,\constM\cdot T]$,
where $\constM=\Theta(1)$  and $T = \poly(n)$ are protocol parameters.
The initial value of $\clock$ is set to $0$ for all agents, including the leader.\footnote{Originally in \cite{AAE08}, the leader’s clock starts at $0$ and all others at $-1$. We instead initialize $\clock$ to $0$ for every agent to simplify the protocol description; this change does not affect the correctness of Lemma~\ref{lemma:phaseclock}.}
When an initiator $\ini$ and responder $\res$ interacts, $\ini$ does nothing and $\res$ updates its clock as follows:
\begin{align*}
\res.\clock \gets 
\begin{cases}
\min(\res.\clock + 1,\constM \cdot T) & \tif \left (
\begin{aligned}
&\res.\leader = 1 \\
&\land \ini.\clock = \res.\clock
\end{aligned}
\right )\\
\max(\ini.\clock,\res.\clock)& \totherwise,
\end{cases}
\end{align*}
where $\leader$ is an indicator variable with $a.\leader = 1$ if and only if $a$ is the unique leader.
We denote this protocol by $\phaseclock(T,\constM)$.

We partition the interval $[-1,\constM\cdot T]$ into \emph{phases} of uniform length $\constM=\Theta(1)$. Specifically, for each $i \in [-1,T]$, we say agent $a$ is in phase $i$ if $\constM\cdot i \le a.\clock < \constM\cdot (i+1)$. For each $i\in[0,T]$, let $\cphase(i)$ denote the set of configurations in which all agents are in phase $i$.
\begin{lemma}[Immediate consequence of Lemma~1 and Theorem~1 in \cite{AAE08}]
\label{lemma:phaseclock}
Let $C_0$ be the configuration in which every agent has $\clock=0$. For any $T=\poly(n)$ and any constant $\cmin=\Theta(1)$, there exists a constant $\constM$ such that, with high probability, an execution of $\phaseclock(T,\constM)$ starting from $C_0$ simultaneously satisfies for all $i\in[1,T]$:
\begin{itemize}
  \item the population enters $\cphase(i)$ within $O(i\cdot n\log n)$ interactions, and
  \item after entering $\cphase(i)$, remains there for at least $\cmin\cdot n\lg n$ interactions.
\end{itemize}
\end{lemma}
\noindent
Note that $\clock$ is non-decreasing in this protocol.  
Hence, by Lemma~\ref{lemma:phaseclock}, the population enters $\cphase(1)$, $\cphase(2)$, $\dots$, $\cphase(T)$ in this order w.h.p.

\begin{algorithm}[t]
\caption{$\ssrk(\rho)$ at an interaction where initiator $\ini$ and responder $\res$ meet.
}
\label{al:main}
\Variables{
$\rank \in [1,n]$,
$\clock \in [0,\constM \cdot \tmax]$,
$\mode \in \{\bot,\modeF,\modeD,\modeR\}$,\\
\hspace{1.63cm}
$\target \in [1,n]$,
$\detected\in \{0,1\}$,\\
\hspace{1.63cm}
$\rst \in [0,\constR \cdot \logn]$,
$\delay \in [0,\constD\cdot \logn]$
}
\Notation{
$\phase(a) = \lfloor a.\clock/\constM \rfloor$
}
Execute $\lsle(\constL)$
\violet{\tcp*{$\constL=\Theta(1)$ is a sufficiently large constant}}
$\ini.\rst \gets \res.\rst \gets \max(\ini.\rst-1,\res.\rst-1,0)$\;
\uIf 
{$\res.\rst > 0$}{
$\res.\clock \gets \ini.\clock \gets 0$ \violet{\tcp*{Initialize $\phaseclock()$}}
$\res.\delay \gets \ini.\delay \gets \constD \cdot \logn$\;
$\res.\mode \gets \ini.\mode \gets \bot$;\ $\res.\detected \gets \ini.\detected \gets 0$\;
}
\uElseIf{$\res.\leader = 1 \land \res.\delay >0$}{
$\res.\delay \gets \res.\delay -1$\;
}
\Else{ 
Execute $\phaseclock(\tmax,\constM)$ \violet{\tcp*{Results in $\phase(\ini) \le \phase(\res)$}}
\If{$\phase(\res)=\tmax$}{
$\res.\rst \gets \constR\cdot \logn$
}
Let $\pini = \phase(\ini),\ \pres = \phase(\res)$\;
Let $\rini = \ini.\target,\ \rres = \res.\target, \dres = \res.\detected$\;
$\res.\mode \gets
\begin{cases}
\modeF & \tif \pres = T_1-1\\
\modeD & \tif \pres = T_2-1\\
\modeR & \tif \pres = T_3-1 \land \dres = 1\\
\res.\mode & \totherwise
\end{cases}
$
\violet{\tcp*{Auto-initialize variables upon mode change}}
\If(\violet{\tcp*[f]{$\pini \le \pres$}}){$\ini.\mode = \res.\mode \land (\res.\mode=\modeD \to \rini=\rres)$}{
Execute
$
\begin{cases}
\findtarget& \tif [\pini,\pres] \subseteq [T_1,T_2-2] \land \res.\mode = \modeF\\
\ccd(r_2,\rho) & \tif [\pini,\pres] \subseteq [T_2,T_3-2] \land \res.\mode = \modeD\\
\ranking& \tif [\pini,\pres] \subseteq [T_3,T_4-2] \land \res.\mode = \modeR\\
\text{nothing}& \totherwise.
\end{cases}
$
}
}
\end{algorithm}

\begin{table}[t]
\caption{
Initialization of each component
}
\label{table:init}
\centering
\begin{tabular}{c c}
\hline
 component & initialization 
 \\ 
 \hline
 \findtarget&
 $a.\target = 1$, 
$a.\parity = (T_1 \bmod{2})$
 \\
$\ccd$ &$a.\detected = 0$, $a.\vlist = \emptyset$, $a.\susp = 0$
 \\
 $\ranking$&
 \begin{tabular}{c}
 $a.\ind = 0$, 
 $a.\rank =
\begin{cases}
1& \tif a.\leader = 1\\
\bot& \totherwise
\end{cases}
$,\\
$a.\nonce = 0$, $a.\parity = (\tmid \bmod 2)$
\end{tabular}
 \\
\hline
\end{tabular} 
\end{table}

\section{Cyclic Execution of Sub-protocols}
\label{sec:cycle}
Using the phase clock~\cite{AAE08} described in Section~\ref{sec:phase_clock}, $\ssrk(\rho)$ invokes the sub-protocols $\findtarget$, $\ccd(r,\rho)$, and $\ranking$ in phases $[T_1,T_2-2]$, $[T_2,T_3-2]$, and $[T_3,T_4-2]$, respectively (line~24), where
$$
T_1 = 2,\quad
T_2 = T_1 + \constT \left\lceil \sqrt{n}\ \right\rceil,\quad
T_3 = T_2 + \constT \left\lceil \frac{n \lg \rho }{\rho \log n} \right\rceil,\quad
T_4 = T_3 + 2\constT \left\lceil \sqrt{n} \right\rceil,
$$
and $\constT=\Theta(1)$ is a sufficiently large constant to ensure the correctness of each component.
If an agent enters phase $\tmax$, it raises a reset flag (lines~18--19), returning all agents to the initial phase $0$ by the epidemic protocol (lines 9--13). Since each phase completes in $\Theta(n\log n)$ interactions w.h.p.\ by Lemma~\ref{lemma:phaseclock}, one cycle of $\tmax$ phases requires only $\Theta(\goalI)$ interactions w.h.p.

The variable $\mode \in \{\bot,\modeF,\modeD,\modeR\}$ is key to self-stabilization.
The modes $\modeF$, $\modeD$, and $\modeR$ correspond to the sub-protocols $\findtarget$, $\ccd$, and $\ranking$, respectively. When an agent switches its $\mode$ from $x$ to a different mode $y \neq \bot$, it runs the initialization procedure for component $y$.
The initialization procedures for $\ccd$, $\ranking$, and $\findtarget$ are described in Sections~\ref{sec:ccd}, \ref{sec:ranking}, and~\ref{sec:findtarget}, respectively.
They are also specified in Algorithms~\ref{al:ccd},~\ref{al:ranking}, and~\ref{al:findtarget} and summarized in Table \ref{table:init}.

Note that $\ccd(r,\rho)$ requires the target rank $r$; upon entering mode $\modeD$, each agent $a$ adopts $a.\target$ for $r$.
Mode switching occurs in phases $T_1-1$, $T_2-1$, and $T_3-1$, i.e., immediately before each component starts (line~22). An agent $a$ enters mode $\modeR$ only if $a.\detected=1$, i.e., its rank-collision-detected flag is raised,
so rank reassignment cannot begin until $\ccd(r,\rho)$ detects a collision. Each component is executed only when both interacting agents share the same mode; additionally, during mode $\modeD$ (collision detection), they must have the same target rank (line~23).

The first part of the main function (lines 9–19) brings all agents into phase~0, resetting their modes to $\bot$, within $O(\goalI)$ interactions w.h.p., regardless of the initial configuration (Lemma~\ref{lemma:reset}).
By Lemma~\ref{lemma:phaseclock}, this ensures the population completes one full cycle of the three components in the next $\Theta(\goalI)$ interactions w.h.p.
This reset mechanism relies on two timers, $\rst\in[0,\constR\cdot \logn]$ and $\delay\in[0,\constD\cdot\logn]$, where a constant $\constD=\Theta(1)$ is chosen sufficiently larger than $\constR=\Theta(1)$.
We omit the detailed explanation of the update rules for $\rst$ and $\delay$ (lines 9–19) in the main text, 
as they follow standard techniques~\cite{BCC+21,SEIM21}, in particular the \protocol{PropagateReset} subprotocol of~\cite{BCC+21}; 
their precise definitions appear in the pseudocode.

\begin{definition}
\label{def:initialized_configurations}
Let $\calI$ be the set of configurations in which
$a.\clock = a.\rst = a.\delay = a.\detected=0$ and $a.\mode = \bot$ hold for all $a \in A$.
\end{definition}

\begin{lemma}
\label{lemma:reset}
Starting from any configuration, the population under $\ssrk(\rho)$ enters $\calI$ within $O(\tmax\cdot n\log n)$ interactions w.h.p. 
\end{lemma}
\begin{proof}
Follows from Lemma 1 of \cite{ADK+17}, Lemma 5 of \cite{SOK+20}, and Lemma~\ref{lemma:phaseclock}.
\end{proof}

For the subsequent analysis, we introduce the following sets of configurations.
\begin{definition}
For each $X \in \{\modeF,\modeD,\modeR\}$, let $\cinit(X)$ denote the set of configurations of $\ssrk(\rho)$ in which every agent
\begin{enumerate}
  \item satisfies $\rst = 0 \land (\leader = 1 \to \delay = 0)$,
  \item is in mode $X$ and in phase $T_1 - 1$, $T_2 - 1$, or $T_3 - 1$ for $X = \modeF, \modeD, \modeR$, respectively, and
  \item has all component-specific variables initialized to their specified initial values.
  (See Table~\ref{table:init} for the specified initial values.)
\end{enumerate}
Moreover, for any $r \in [1,n]$, let $\cinit(\modeD,r)$ denote the set of configurations in $\cinit(\modeD)$ in which $a.\target = r$ for every agent $a \in A$.
\end{definition}

\begin{algorithm}[t]
\caption{$\ranking$ at an interaction where initiator $\ini$ and responder $\res$ meet.
} 
\label{al:ranking}
\Notation{
$\rem = \lceil \sqrt{n} \rceil$,\ 
$\phase(a) = \lfloor a.\clock/\constM \rfloor$
}
\Variables{
$\ind \in [0,\logn]$, 
$\nonce \in [0,n^2]$, 
$\cand, \parity \in \{0,1\}$
}
\Initially{
$a.\ind = 0$, $a.\rank =
\begin{cases}
1& \tif a.\leader = 1\\
\bot& \totherwise
\end{cases}
$\\
\hspace{1.68cm}
$a.\nonce = 0$, $a.\parity = (\tmid \bmod 2)$
}
\Remark{
\violet{
Within this subprotocol, $\rank\in[1,n]\cup\{\bot\}$ (see Note~\ref{note:domain}).
}
}
Let $\pini = \phase(\ini),\ \pres = \phase(\res)$ \violet{\tcp*{$\pini \le \pres$}}
\If{$\pres < \tmid$}{
\If{$\ini.\rank \neq \bot  \land \ini.\ind < \logn \land \res.\rank = \bot$}{
$\ini.\ind \gets \ini.\ind + 1$\;
Let $\rnew = \ini.\rank + 2^{\logn - \ini.\ind}$\;
\If{$\rnew \le n-\rem$}{
$\res.\rank \gets \rnew$\;
$\res.\ind \gets \ini.\ind$
}
}
}
\If{
$\pres\ge \tmid$ and $\pres \equiv \res.\parity \pmod{2}$
}{
$\res.\parity \gets 1-\res.\parity$\;
Let $h=\pres-\tmid$\;
\If{$h \equiv 0 \pmod{\rept}$}{
 \If{$\res.\cand=1$}{
 $\res.\rank \gets n-\rem-1+h/\rept$
 } 
 $\res.\cand \gets
 \begin{cases}
 1& \tif \res.\rank = \bot \\
 0& \totherwise
 \end{cases}
 $
 }
 $\res.\nonce \gets 
 \begin{cases}
 \text{an integer chosen u.a.r.~from }[1,n^2] & \tif \res.\cand = 1\\
 0 & \totherwise
 \end{cases}
 $\; 
 \violet{\tcp*{Can be derandomized. See Note \ref{note:derandomize}}}
}
\If{$\pini=\pres \ge \tmid \land \res.\nonce < \ini.\nonce$}{
$\res.\cand \gets 0$\;
$\res.\nonce \gets \ini.\nonce$\;
}
\end{algorithm}

\section{Sublinear-Time, Polynomial-State Ranking}
\label{sec:ranking}
We present a ranking protocol $\ranking$ that assigns unique ranks in $[1,n]$ within $O\left(\rtimeI\right)$ interactions (i.e., $O\left(\rtimeT\right)$ time) w.h.p.\ using only polynomially many states, starting from any configuration in $\cinit(\modeR)$. Throughout this subsection, let 
$\rem = \left\lceil \sqrt{n} \right\rceil$.

\begin{note}
\label{note:domain}
Within this protocol, $\rank$ ranges over $[1,n]\cup\{\bot\}$ rather than $[1,n]$. If $a.\mode$ leaves $\modeR$ while $a.\rank=\bot$, we reset $a.\rank$ to some value in $[1,n]$ (e.g.\ $1$).
\end{note}

The module $\ranking$, shown in Algorithm~\ref{al:ranking}, consists of two parts.  
In the first part (lines 25--32), it assigns $n-\rem$ agents ranks in $[1,n-\rem]$ \emph{in parallel}.  
In the second part (lines 33--44), it assigns the remaining $\rem$ agents ranks in $[n-\rem+1,n]$ \emph{sequentially}.  
The first part requires $O\left(\sfrac{n^2}{\rem} \cdot \log n\right)$ interactions w.h.p., and the second requires $O(\rem n\log n)$ interactions w.h.p.  With $\rem=\sqrt{n}$, the total becomes $O\left(\rtimeI\right)$, as claimed.  

Recall that this subprotocol runs over phases $[T_3,T_4-2]$, where 
$T_4 = T_3 + 2\constT \left\lceil \sqrt{n} \right\rceil$.
To execute its two parts sequentially, we split these phases nearly in half into 
$[T_3,\tmid-1]$ and $[\tmid,T_4-2]$, where
$$
\tmid = T_4 - 1 - \constT \left\lceil \sqrt{n} \right\rceil.
$$
Note that the second interval consists of exactly $\constT \left\lceil \sqrt{n} \right\rceil$ phases.

We design the first part, inspired by the $O(n)$-time, $O(n)$-state SS-RK protocol \protocol{Optimal-Silent-SSR} of~\cite{BCC+21}. Whereas \protocol{Optimal-Silent-SSR} assigns all $n$ ranks, our first part assigns only $n-\rem$ ranks, deferring the remaining $\rem$ ranks to the second part. This reduces the interaction complexity from $O(n^2)$ in expectation (and $O(n^2 \log n)$ w.h.p.) to $O(n^{3/2}\log n)$ interactions w.h.p.

The first part assigns $n-\rem$ ranks using a variable $\ind\in[0,\logn]$. We call an agent with $\rank=\bot$ a \emph{null} agent.
Each non-null agent $a$ interprets the pair $(a.\rank-1,\,a.\ind)$ as a $\logn$-bit integer whose first $a.\ind$ bits are fixed and whose remaining bits are masked:
$$
\underbrace{b_1 b_2 \cdots b_{a.\ind}}_{a.\ind\text{ bits}}***\cdots*,
$$
where `*' denotes a masked bit. Note that subtracting 1 from $a.\rank$ maps ranks in $[1,n-\rem]$ to the zero‐based range $[0,n-\rem-1]$.
When $a.\ind=\logn$, $a.\rank$ is fully determined and no longer changes. 
Upon switching to mode $\modeR$, agents use the leader $a_L$, maintained by module $\lsle$ (line~8): the unique leader $a_L$ sets $\rank=1$, while all others set $\rank=\bot$. All agents also initialize $\ind=0$. Hence initially only $a_L$ is non-null with $(\rank-1,\ind)=(0,0)$, and all others are null. 
The first part proceeds when a non‐null initiator $\ini$ with $\ini.\ind<\logn$ meets a null responder $\res$. Then $\ini.\ind$ increments by one (line~28)---replacing the leftmost masked bit with $0$---and $\rnew-1$ is obtained by setting that bit to $1$ in the binary representation of $\ini.\rank-1$ (line~29). If $\rnew\le n-\rem$, the responder sets $(\rank,\ind)=(\rnew,\ini.\ind)$ (lines~30--32).

In the first part, at most $n-\rem$ ranks are generated, so at least $\rem$ null agents always remain. This key property yields faster stabilization. Fix any $r\in[1,n-\rem]$. At each time step $t$, there is exactly one agent $a_r(t)\in A$ satisfying
$$
\lfloor (r-1)/p\rfloor \cdot p = a_r(t).\rank-1,\quad p=2^{\logn - a_r(t).\ind}.
$$
Each interaction in which $a_r(t)$ (as initiator) meets a null agent increments $a_r(t).\ind$ by one; after $\logn$ such meetings, some agent has $(\rank,\ind)=(r,\logn)$, completing rank assignment for $r$. Since this event occurs with probability at least $\rem/(n(n-1))=\Omega(\rem/n^2)$ per step, by the Chernoff bound the assignment of $r$ completes in $O\left(\sfrac{n^2}{\rem}\cdot \log n\right)$ interactions w.h.p. By a union bound over all $r\in[1,n-\rem]$, the first part finishes in $O\left(\sfrac{n^2}{\rem}\cdot \log n\right)=O(n^{3/2}\log n)$ interactions w.h.p., as desired.

The second part assigns $\rem$ ranks sequentially, using the phase clock that $\ssrk(\rho)$ invokes outside this module (line~17). In each round of $\constT$ phases, all null agents contend to win a new rank. With high probability, exactly one agent wins per round and obtains its assigned rank. This follows a standard leader‐election technique in population protocols~\cite{MST18,SOI+20,GJLL25}. Specifically, when an agent first reaches phase $p\in[\tmid,T_4-2]$—detectable via a binary variable $\parity\in\{0,1\}$—it executes:
\begin{itemize}
  \item If $p-\tmid$ is divisible by $\constT$, a null agent sets its $\cand$ flag to $1$ (line~39), becoming a \emph{candidate} for rank $n-\rem+1+(p-\tmid)/\constT$. A candidate that retains $\cand=1$ for the next $\constT$ phases wins that rank; non‐null agents always reset $\cand=0$.
  \item Every agent with $\cand=1$ chooses its nonce u.a.r.\  from $[1,n^2]$ (derandomizable—see Note~\ref{note:derandomize}); agents with $\cand=0$ set $\nonce=0$.
\end{itemize}
In each phase, the maximum nonce propagates to the entire population w.h.p.\ via the epidemic protocol among agents in the same phase (lines~42--44). Any agent that observes a larger nonce resets its $\cand$ flag to $0$ (line~43).

Let $\ncand$ be the number of candidates in a given phase. They draw nonces $x_1,\dots,x_{\ncand}$ u.a.r.\ from $[1,n^2]$. For each $i\in[2,\ncand]$, 
$\Pr[x_i = \max\{x_1,\dots,x_{i-1}\}]=1/n^2$, 
so by a union bound over $i=2,\dots,\ncand$, there is a unique maximum nonce with probability 
$1 - O(\ncand/n^2) = 1 - O(1/n)$. Hence, by choosing $\constT=\Theta(1)$ sufficiently large, each round selects exactly one winner w.h.p.

By the above discussion, we have the following lemma.
\begin{lemma}
\label{lemma:ranking}
Starting from any configuration $C \in \cinit(\modeR)$, the population under $\ssrk(\rho)$ reaches a configuration where all agents have distinct ranks in $[1,n]$ within $O(\rtimeI)$ interactions w.h.p.
\end{lemma}

\begin{note}[Derandomization]
\label{note:derandomize}
In $\ranking$, each non-null agent generates a random number from $[1,n^2]$, making the protocol randomized. However, we can derandomize it by exploiting the scheduler’s randomness: whenever agent $a$ is chosen for an interaction, it becomes initiator with probability $1/2$ (and responder otherwise), yielding one random bit per interaction. By Lemma~\ref{lemma:phaseclock}, each phase $i$ involves $\Theta(n\log n)$ interactions w.h.p., of which $a$ participates in $\Theta(\log n)$ interactions w.h.p.\ (for sufficiently large constant $\cmin$). Thus $a$ obtains $\Theta(\log n)$ random bits per phase—enough to generate a number in $[1,n^2]$ for the next phase. This assumes $n^2$ is a power of 2, but replacing $n^2$ with $2^{\lceil2\lg n\rceil}$ only increases the success probability.
Note that this derandomization requires a multiplicative factor of $n^2$ in the number of states, although this does not affect our results: it suffices for $\ranking$ to use only polynomially many states. Moreover, if one wishes to eliminate these $n^2$ states, a special phase can be periodically inserted that is dedicated to generating random bits, where agents reuse the variable $\nonce$ to produce $2\log n$ random bits.
Note also that this derandomization does not make nonces independent across agents. In $\ranking$, such dependence only accelerates the elimination of losing candidates. If full independence is required, one can use the derandomization method of \cite{BCC+21} (see Section~6 of the arXiv version~\cite{BCC+21arxiv}).
\end{note}

\section{Finding a Target}
\label{sec:findtarget}
In this section, we design the subprotocol $\findtarget$ by slightly modifying the collision-detection protocol $\coldbfull$ of~\cite{AS26}, abbreviated $\coldb$.
Protocol $\coldb$ detects rank collisions in $O(n^{3/2} \cdot \sqrt{\log n})$ expected interactions using $O(n\cdot\mathrm{poly}(\log n))$ states, and never reports a false positive when starting from a correctly initialized configuration. We modify $\coldb$ so that, if a \emph{duplicate} rank (i.e., a rank shared by two or more agents) exists, all agents agree on one such rank and store it in $\target$ within $O(n^{3/2} \log n)$ interactions with probability $1-o(1)$. If no duplicate exists, $\findtarget$ may set $\target$ arbitrarily (we default to 1).

Protocol $\coldb$ relies on two assumptions, both of which are satisfied in our setting. First, it requires approximate knowledge of the population size $n$ (asymptotically tight lower and upper bounds), whereas we provide exact knowledge of $n$. Second, it assumes a unique leader to operate a single-leader phase clock~\cite{AAE08}; we maintain such a leader $a_L$ via the LS-LE protocol $\lsle$ (line~8) and run the phase clock (line~17) (see Section~\ref{sec:tools}), allowing the computed phase to be reused by $\coldb$.

Hence, we obtain $\findtarget$ as follows.  Upon switching to mode $\modeF$, each agent sets $\target=1$ (and resets auxiliary variable $\parity$ to $T_1 \bmod{2}$.) Thereafter, on each interaction they execute $\coldb$.  If a collision is detected, they compute a duplicate rank $r$ and overwrite $\target$ with $r$.  In phase $T_2-2$, i.e., the final phase for this sub-protocol, agents stop running $\coldb$ and instead execute an epidemic on $\target$, so that the maximum duplicate rank found so far propagates to the entire population.
For completeness, we include the pseudocode of $\findtarget$ in the appendix (Algorithm~\ref{al:findtarget}).

We now lower‐bound the success probability of $\findtarget$ by $1-o(1)$. \cite{AS26} proved that for any integer $k\in[1,\Theta(\log n)]$, $\coldb$ detects a collision (if exists) within $O(k\cdot n^{3/2}\cdot\sqrt{\log n})$ interactions with probability $1-(2/3)^k$. (See the proof of Lemma 5 in \cite{AS26}.) Since agents remain in phases $[T_1,T_2-3]$ for $\Theta(n^{3/2}\cdot\log n)$ interactions, choosing $k=\sqrt{\log n}$ gives success probability $1-o(1)$.

By the above discussion, we have the following lemma.

\begin{lemma}
\label{lemma:target}
For any configuration $C$, let $R(C)$ denote the set of duplicate ranks appearing in $C$.  
Starting from any $C \in \cinit(\modeF)$ with $|R(C)| \ge 1$, the population under $\ssrk(\rho)$ reaches, within $O(n^{3/2}\log n)$ interactions and with probability $1 - o(1)$, a configuration in which every agent stores the same value $r \in R(C)$ in the variable $\target$.
\end{lemma}

\section{Proof of Main Theorem}
\label{sec:main_theorem}

\begin{proof}[Proof of Theorem~\ref{theorem:main}]
We first bound the number of states.  In $\ssrk(\rho)$, the largest domain arises from the variable $\vlist$ in $\ccd(r,\rho)$, which ranges over all subsets of $\rname=[1,\mname]\setminus\{r\}$ of size at most $\maxnum$.  Hence the size of the domain is at most 
$$
\sum_{i=0}^{\maxnum} \binom{\mname}{i}\le \sum_{i=0}^{\maxnum} (\mname)^i=O(\rho^{2\maxnum})=O(2^{2\rho\lg^2\rho}).
$$
All other variables range over polynomially sized domain.
Thus, the total number of states is 
$O(2^{2\rho \lg^2 \rho}) \cdot \poly(n) = 2^{2\rho\lg^2 \rho + O(\log n)}$.

Next, we bound the stabilization time.  
By Lemma~\ref{lemma:lsle}, for any $T = \poly(n)$, starting from any configuration the population reaches, within $O(n\log n)$ expected interactions, a configuration $C$ from which a unique leader is preserved for at least $T$ interactions w.h.p.
Therefore, by Lemma~\ref{lemma:closure}, it suffices to show that, starting from any such configuration $C$, the population reaches a safe configuration within $O(\goalI)$ interactions with probability $\Omega(1)$ (see Definition~\ref{def:safe} for the definition of safe configurations); the expected-time bound then follows.  
(Note that Lemma~\ref{lemma:closure} does not require the unique-leader assumption.)

By Lemma~\ref{lemma:reset}, within $O(\goalI)$ interactions w.h.p.\ the system reaches some $C_0\in\calI$.  If $C_0$ has no duplicate rank, then $\detected=0$ and $\mode=\bot$ in $C_0$, so $C_0$ is safe and we are done.  Otherwise, suppose $C_0$ contains at least one duplicate rank.

By Lemma~\ref{lemma:phaseclock}, within $O(n\log n)$ interactions w.h.p.\ the execution enters $\cinit(\modeF)$ from $C_0 \in \calI$.  Then by Lemma~\ref{lemma:target}, within $O(n^{3/2}\log n)$ interactions with probability $1-o(1)$, $\findtarget$ causes all agents to agree on a duplicate rank $r$, reaching $\cinit(\modeD,r)$.  From there, Lemma~\ref{lemma:detect} implies that within $O(\goalI)$ interactions with probability $\Omega(1)$ every agent raises $\detected=1$ and the execution enters $\cinit(\modeR)$.  Finally, by Lemma~\ref{lemma:ranking}, within $O(n^{3/2}\log n)$ interactions w.h.p.\ the population assigns all $n$ agents distinct ranks, and by Lemma~\ref{lemma:phaseclock} another $O(n^{3/2}\log n)$ interactions suffice to reach phase~$\tmax$. Then, all agents reset $\mode$ and $\detected$ by epidemic for $\rst$, thus reaching a safe configuration.  

By a union bound over these phases, the probability of reaching a safe configuration within $O(\goalI)$ interactions is still $\Omega(1)$.  This completes the proof.
\end{proof}

\section{Conclusion}
\label{sec:conclusion}

In this paper, we have presented a new time--space tradeoff for self-stabilizing leader election (SS-LE) in the population protocol model under exact knowledge of the population size $n$. Independently of the recent protocol of~\cite{ABF+25}, we showed:
\begin{itemize}
  \item For any integer $2\le \rho\le \sqrt{n}$, there exists an SS-LE protocol that stabilizes in $O(\sfrac{n}{\rho}\cdot\log\rho)$ expected time using $2^{2\rho\lg^2\rho+O(\log n)}$ states.
  \item By choosing $\rho=\Theta\left(\sfrac{\log n}{\log^2\log n}\right)$, this yields the first \emph{polynomial}-state protocol with sublinear expected time.  
\end{itemize}
Compared to the protocol of~\cite{ABF+25}, ours uses drastically fewer states whenever the target expected time is at least $\Theta(\sqrt{n}\log n)$.

Our result complements prior work by filling the gap between the linear-state, linear-time regime and the super-exponential-state, optimal-time regime.  A remaining limitation is that the restriction $\rho\le\sqrt{n}$ prevents us from achieving $o(\sqrt{n}\log n)$ expected time; by contrast, the protocol of~\cite{ABF+25} can reach the optimal $\Theta(\log n)$ time at the cost of super-exponential state size.

The following questions remain open:
\begin{itemize}
  \item Does there exist an SS-LE protocol using only polynomially many states that stabilizes in $O(n^{1-\epsilon})$ time for some constant $\epsilon>0$?  Our time--space tradeoff achieves only $O\left(n/\poly(\log n)\right)$ time under polynomial‐state constraints.
  \item Does Conjecture~\ref{conjecture:ss-rk} hold?  Equivalently, must one solve SS-RK in order to achieve SS-LE?
\end{itemize}

\paragraph*{Acknowledgments}
This work is supported by JST FOREST Program JPMJFR226U and 
JSPS KAKENHI Grant Numbers JP20KK0232, JP25K03078, and JP25K03079.

\bibliographystyle{alpha}
\bibliography{population}

\clearpage

\appendix

\begin{algorithm}[t]
\caption{$\findtarget$ at an interaction where
initiator $\ini$ and responder $\res$ meet.
} 
\label{al:findtarget}
\Notation{
$r = \lceil \lg n \rceil$, $\ell = \left \lceil \sqrt{n\log n} \right \rceil$, and $z = \lceil n / \ell \rceil$, $\phase(a) = \lfloor a.\clock/\constM \rfloor$
}
\Variables{
$\target \in [1,n]$,
$\grid \in ([1,\ell]\times\{0,1\})\cup \{\bot\}$,
$\infectivity \in [0,\lfloor \lg n - \lg \ell \rfloor]$,
$\parity \in \{0,1\}$
}
\Initially{
$a.\target = 1$, 
$a.\parity = (T_1 \bmod{2})$
}
Let $\pini = \phase(\ini),\ \pres = \phase(\res)$ \violet{\tcp*{$\pini \le \pres$}}
\uIf{$\pres < T_2 -2$}{
\If{$\ini.\rank = \res.\rank$}{
$\res.\target \gets \res.\rank$\;
}
Let $k=(\pres - T_1)\bmod z$\;
\If{$\pres \equiv \res.\parity \bmod{2}$}{
$\res.\parity \gets 1 - \res.\parity$\;
\uIf{$\res.\rank \in [k\cdot \ell+1,(k+1)\cdot \ell]$}{
Choose $\chi$ uniformly at random from $\{0,1\}$\violet{\tcp*{Can be derandomized.}}
$\res.\grid \gets (\res.\rank - k\cdot \ell,\chi)$\;
$\res.\infectivity \gets \lfloor \lg n - \lg \ell \rfloor$\;
}\Else{
$\res.\grid \gets \bot$\;
$\res.\infectivity \gets 0$\;
}
}
\If{
$\pini = \pres$
}{
\If{$\ini.\grid \neq \bot \land \ini.\infectivity > 0 \land \res.\grid = \bot$}{
$\res.\grid \gets \ini.\grid$\;
$\ini.\infectivity \gets \res.\infectivity \gets \ini.\infectivity-1$
}
\If{$\exists \alpha,\beta,\gamma: \ini.\grid=(\alpha,\beta) \land \res.\grid=(\alpha,\gamma) \land \beta \neq \gamma$}{
$\res.\target \gets \min(n,k\cdot \ell+\alpha)$\;
}
}
}\Else{
$\res.\target \gets \max(\res.\target,\ini.\target)$
}
\end{algorithm}

\section{Omitted Proof and Pseudocode}
\label{sec:proofs}
\begin{proof}[Proof of Remark \ref{remark:lower_bound_abf}]
In $\assignranks_{\rho}$, the $n$ agents are eventually split into $\rho$ \emph{deputies} and $n-\rho>n/2$ \emph{recipients}.  Before assigning unique ranks to all $n$ agents, each recipient must meet at least one deputy.  Therefore, the expected number of interactions to assign distinct ranks to all agents is at least
$$
\sum_{i=1}^{n-\rho}\frac{\binom{n}{2}}{i\,\rho}
=\Omega\left(\tfrac{n^2\log n}{\rho}\right),
$$
i.e., $\Omega(\sfrac{n}{\rho}\cdot\log n)$ expected time.
\end{proof}

\begin{proof}[Proof of Lemma \ref{lemma:lsle}]
\cite{SEIM21} prove that for any constant $\tau \ge 1$,
there exists a set $\csafe$ of configurations of $\lsle(\tau)$ such that
(i) starting from an arbitrary configuration, the population under $\lsle(\tau)$ reaches a configuration in $\csafe$ within $O(\tau n\log n)$ expected interactions, and (ii) starting from any configuration in $\csafe$,
the population goes back to a safe configuration after $\Theta(\tau n \log n)$ interactions and retains a unique leader during the period with probability $1-n^{-\tau}$ (See the proof of Lemma 15 in Appendix A.1 of \cite{SEIM21} for details).
The lemma immediately follows from this fact and the union bound, adjusting $\tau = \Theta(1)$ sufficiently.
\end{proof}

\end{document}